\documentclass[sn-mathphys,Numbered]{sn-jnl}

\usepackage{graphicx}%
\usepackage{multirow}%
\usepackage{amsmath,amssymb,amsfonts}%
\usepackage{amsthm}%
\usepackage{mathrsfs}%
\usepackage[title]{appendix}%
\usepackage{xcolor}%
\usepackage{textcomp}%
\usepackage{manyfoot}%
\usepackage{booktabs}%
\usepackage{algpseudocode}%
\usepackage{listings}%
\usepackage{caption} 
\usepackage{subcaption}
\usepackage{amsmath,amsfonts}
\usepackage{amsthm,amssymb}
\usepackage{float}

\theoremstyle{thmstyleone}%
\newtheorem{theorem}{Theorem}
\newtheorem{proposition}[theorem]{Proposition}%

\theoremstyle{thmstyletwo}%
\newtheorem{remark}{Remark}%

\theoremstyle{thmstylethree}%

\begin{document}

\title[]{Selective Forgetting in Option Calibration: An Operator-Theoretic Gauss--Newton Framework}


\author*[1 ]{\fnm{Ahmet Umur} \sur{\"Ozsoy}}\email{umurozsoy@gmail.com}

\affil[1]{\orgdiv{Department of Industrial Engineering}, \orgname{Istanbul Okan  University}, \orgaddress{  \city{Istanbul}, \postcode{34959}, \country{Turkey}}}


\abstract{Calibration of option pricing models is routinely repeated as markets evolve, yet modern systems lack an operator for removing data from a calibrated model without full retraining. When quotes become stale, corrupted, or subject to deletion requirements, existing calibration pipelines must rebuild the entire nonlinear least-squares problem, even if only a small subset of data must be excluded. In this work, we introduce a principled framework for selective forgetting (machine unlearning) in parametric option calibration. We provide stability guarantees, perturbation bounds, and show that the proposed operators satisfy local exactness under standard regularity assumptions. }

\keywords{selective forgetting, option calibration, Gauss--Newton methods, 
	sufficient statistics, numerical unlearning, Heston model}



\maketitle

\section{Introduction}


Modern financial models are not static; they are  recalibrated as market conditions change.
Therefore calibrating parametric asset-pricing models to market data has always been an ongoing interest for both practitioners and academics in the field of mathematical finance.
Risk management systems along with trading desks rely heavily on the repeated solutions of inverse problems aimed at calibrating  and  adjusting parameters $\theta$ so that the model-based prices $m(x;\theta)$ reproduce  observed quotes  to some extent of accuracy.
Option-implied volatility surfaces evolve minute by minute, and model parameters such as
mean reversion, volatility of volatility, or correlation etc. are adapted to new market information.
Formally, calibration seeks parameters $\theta$ minimizing a discrepancy between model generated prices
$m(x;\theta)$ and observed quotes, typically through nonlinear least squares or
maximum-likelihood estimation. 
This sort of inverse problem is present at models such as Heston and SABR
to structural credit, interest rate, and hybrid models and  lies at almost all the operational core of
risk engines and trading platforms.

While long-standing research have refined the \emph{estimating} or \emph{learning} side of calibration of financial derivatives, options in our case, no interest has been shown to its conceptual dual, \emph{unlearning}.
When certain data becomes corrupted, obsolete, or subject to possible deletion requests,
the model should exclude its influence without a full recalibration from scratch.
Therefore deletion and data retraction could become necessities as quotes expire, bad ticks or outlier surfaces are purged, and perhaps regulatory obligations (e.g. GDPR, audit requests, or data licensing constraints) require that specific records be removed from already-calibrated models.
The question is immediate yet unsolved, therefore the question becomes how a calibrated model retract the influence of certain data without a full recalibration from scratch.
The removal (i.e., deletion) of certain subsets of available data from an already trained models has emerged recently in machine learning as \emph{Machine Unlearning} or \emph{Selective Forgetting}.
The essence of such approaches casts a simple question whether or not the exclusion of a subset of data requires retraining on the remaining data,~\cite{ginart2019making, bourtoule2021machine, sekhari2021remember,qu2024learn, guo2019certified},
and has become central to privacy-aware and data-efficient algorithm design.

In financial modeling, however, unlearning has not been formalized.
To address this, we introduce a principled framework for \emph{machine unlearning in option calibration}.
Our goal is to update calibrated parameters  as if certain quotes had never been observed,
without re-accessing or reprocessing the entire dataset.
We cast calibration as a nonlinear least-squares problem solved by Gauss--Newton iterations
and show that the normal-equation structure naturally supports machine unlearning.

The term \emph{machine unlearning} has traditionally referred to privacy-motivated deletion in which the goal is to modify a trained model so that it becomes indistinguishable from one that
was never exposed to the deleted data~\cite{ginart2019making}. 
That definition requires a probabilistic or differential-privacy guarantee, an epistemic statement that an observer cannot tell whether the forgotten data ever influenced the model.

Our work adopts the unlearning viewpoint in a different, numerical sense that is particularly natural for calibration. 
We are not concerned with information-theoretic indistinguishability, but with the \emph{computational removal of numerical influence}.
Given a parametric model $m(x;\theta)$ calibrated on market data $D$, and a subset $F\subset D$ of quotes to be discarded (e.g., due to stale prices or data errors), our
objective is to obtain parameters
\[
\theta' \;=\; \arg\min_\theta J(\theta;D\setminus F),
\]
\emph{without reprocessing all of $D$}.  
In this setting, unlearning means reproducing the same parameter update that full retraining on $D\setminus F$ would yield, up to machine precision.
This provides a strict advantage as retraining on large option datasets could be expensive, especially for models with higher complexity.
Hence the \emph{forgetting} is not about privacy or randomness in our point of view, but about \emph{efficiently erasing the numerical footprint} of specific data in parameter space. We refer to~\cite{zhang2023review} and~\cite{nguyen2025survey} for recent and well articulated examples of review of \emph{Machine Unlearning} as the literature is quite vast to provide deep holistic view.


Even though we touch upon the fact that no interest has been shown to the removal of data that compels a recalibration, there are studies that are directionally \emph{forward in time}  in which the models are refit periodically conditional on the accumulation of \emph{new} data, for instance ~\cite{date2011linear},~\cite{broto2004estimation},~\cite{bakshi1997empirical},~\cite{broadie2007model}.  
Such studies perform incremental steps on new samples, by the very nature of it.
By construction, however, they do not address the inverse problem of properly removing the informational contribution of a specific subset of observations while reproducing the solution that would have been obtained had those observations never been used.

Calibration could  be repeated thousands of times as market conditions evolve, often under strict latency and consistency constraints.
This increases our interest in \emph{Machine Unlearning} as a strong reason of removal bulk amount of data might be on cleaning corrupted data.
Even in stable markets, bad ticks, stale or misrecorded quotes could exist and distort calibration.
Therefore unlearning such contaminated shards of data could improve calibration precision without retraining the model from scratch. 
Consider that a pricing engine or feed producing several days of quotes with a decimal-shift bug, instead of full retrain; unlearning those days might realign the parameters with the clean market.
Therefore, especially inspired by the \textsc{SISA} (Sharded Isolated Sliced Aggregation) paradigm from~\cite{bourtoule2021machine}, we suggest two unlearning operators by showing that in nonlinear calibration, the Gauss--Newton equations can be reorganized into algebraically additive terms that admit exact deletion operators as the apparent simplicity hides a structural insight.

This work introduces a principled framework for \emph{selective forgetting} (machine unlearning) in parametric option calibration, formulated under the standard Gauss--Newton (GN) least-squares setting used in Heston-type models. 
Given an initially calibrated parameter vector $\theta$ fitted on a dataset of option quotes $D = \{(x_i,y_i)\}$, the goal is to efficiently update $\theta$ to the parameter that would result from retraining on the retained subset $D \setminus F$, without accessing the full dataset again.
Our contributions are both algorithmic and theoretical.
%
We first design a shard-aware decomposition of the GN normal equations and then we prove that this system coincides exactly with the one obtained by retraining on $D\setminus F$ at a fixed linearization. We term this approach the sharded recompute as it enables machine-precision unlearning with partial data access and provides a scalable, shard-local recalibration architecture for option models.
We then develop a data-free \emph{refactor} operator that realizes exact forgetting without reopening any raw quotes making this the first exact forgetting operator for nonlinear least-squares calibration models, to our knowledge.
The approach achieves retraining-level accuracy while remaining completely data-free once the sufficient statistics $(u_i,\psi_i)$ are cached.

We further  show in   synthetic   option datasets that our framework achieves near-zero degradation in calibration accuracy compared to full retraining, while reducing computational cost by an order of magnitude or more.  
The proposed framework supports operational scenarios in which data must be removed due to regulatory, contractual, or quality-control reasons, providing a principled and efficient alternative to discarding and recalibrating the entire dataset.
By applying selective forgetting to targeted subsets (e.g., quotes from specific dates or sources), we measure their influence on calibrated parameters and pricing accuracy, enabling a form of leave-one-shard-out sensitivity analysis for option pricing models.
Overall, this study bridges the emerging field of machine unlearning with the long-standing problem of derivative model calibration, introducing both a novel theoretical framework and an immediately applicable methodology for real-world financial modeling.
Beyond computational efficiency, the proposed framework positions unlearning  as a fundamental complement to calibration in model management.
In large-scale pricing and risk systems, models must not only learn from new data  but also \emph{forget} obsolete or restricted information.
The machine unlearning operators developed here establish an analytical bridge  between machine unlearning and quantitative finance, enabling the first  operator-theoretic treatment of data deletion in stochastic-volatility model calibration.

\section{Formulation of the unlearning problem}

The unlearning procedure proposed in this study operates purely on the normal equations of the Gauss--Newton method. 
As such, it does not modify the underlying option pricing model, the risk-neutral pricing map, or any no-arbitrage structure inherent to the parametric family.
While the numerical experiments use European calls under the Heston model for analytical clarity,
the proposed forgetting framework extends to any differentiable pricing map,
including exotic or path-dependent contracts evaluated by Monte Carlo or adjoint methods.
The machine unlearning calibration framework developed here does respect all classical no-arbitrage and mathematical finance principles (i.e., monotonicity, convexity, existence of the equivalent martingale measures etc.) to the same extent as the underlying model  calibrated.
In this section, we first discuss preliminaries and then mathematically develop the machine unlearning operators.

\subsection{Preliminaries}

We consider the problem of calibration of the Heston model to European call option prices. 
Under the risk--neutral measure, the Heston dynamics for the asset price $S_t$ and its variance $v_t$ are
\[
\mathrm{d}S_t = rS_t\,\mathrm{d}t + \sqrt{v_t}S_t\,\mathrm{d}W_t^S, \qquad
\mathrm{d}v_t = \kappa(\theta_v - v_t)\,\mathrm{d}t + \sigma_v\sqrt{v_t}\,\mathrm{d}W_t^v,
\]
with $\mathrm{d}\langle W^S,W^v\rangle_t = \rho\,\mathrm{d}t$.
The parameter vector is
$\theta = (\kappa,\theta_v,\sigma_v,\rho,v_0)$.
Call option prices $m(x;\theta)$ are given in semi-closed form via the characteristic function of $\log S_T$~\cite{heston1993closed}.
European call prices under the Heston model admit the semi-analytical representation
\begin{equation}
	\label{eq:heston-price}
	C(S_0,K,T;\theta)
	= S_0 P_1 - K e^{-rT} P_2,
\end{equation}
where the risk--neutral probabilities, \(P_j\), are given by
\begin{equation}
	\label{eq:heston-pj}
	P_j = \tfrac{1}{2} + \tfrac{1}{\pi}
	\int_0^{\infty}
	\Re\!\left(
	\frac{e^{-i u \ln K}\,\phi_j(u)}{i u}
	\right)\mathrm{d}u,
	\qquad j\in\{1,2\}.
\end{equation}
%
%
The characteristic function $\phi_j(u)$ follows the standard form of~\cite{heston1993closed}, depending on parameters $\theta = (\kappa,\theta_v,\sigma_v,\rho,v_0)$.
We evaluate \eqref{eq:heston-pj} numerically via Simpson integration with an upper bound $U_{\max}$ and $N$ sub-intervals, which ensures differentiability of $C(\cdot;\theta)$ with respect to each parameter which ensures that the Jacobian $J_i(\theta)=\nabla_\theta m(x_i;\theta)$ exists for each quote given that $x_i$ $m(x;\theta)$ the parametric pricing map (e.g., Heston) with $\theta \in \Theta \subset \mathbb{R}^p$ with $D=\{(x_i,y_i)\}_{i=1}^N$ denoting the dataset of option quotes (features $x_i$ and responses $y_i$).
Further, let $x = (S,K,T,r)$ denote the market features of a European call option quote,
and let $\theta = (\kappa,\theta_v,\sigma_v,\rho,v_0)$ be the Heston parameter vector.
We define the parametric pricing map
\[
m(x;\theta)
= S\,P_1(S,K,T,r;\theta)
- K e^{-rT} P_2(S,K,T,r;\theta),
\]
where $P_1$ and $P_2$ are the risk--neutral probabilities given by the
Fourier--Laplace integrals of \eqref{eq:heston-pj}.
The calibration problem then seeks
\[
\theta^\star = \arg\min_{\theta} J(\theta;D),
\qquad
J(\theta;D) = \sum_{i\in D} \big(y_i - m(x_i;\theta)\big)^2.
\]
With $r_i(\theta)=y_i-m(x_i;\theta)$, linearizing each residual around a reference $\theta^{\mathrm{ref}}$ gives $r_i(\theta^{\mathrm{ref}}+\Delta\theta) 
\approx r_i(\theta^{\mathrm{ref}}) - J_i(\theta^{\mathrm{ref}})\Delta\theta$, where $J_i(\theta)=\nabla_\theta m(x_i;\theta)$ is the sensitivity (Jacobian) of the model output with respect to parameters.
Substituting into $J(\theta)$ and minimizing the quadratic approximation yields
the \emph{Gauss--Newton normal equations}
\[
H(\theta^{\mathrm{ref}})\,\Delta\theta = G(\theta^{\mathrm{ref}}),
\qquad
H=\sum_i J_i^\top J_i,\quad
G=\sum_i J_i^\top r_i.
\]
Solving for $\Delta\theta$ provides the parameter correction that minimizes the
local linearized loss:
$\theta'=\theta^{\mathrm{ref}}+\Delta\theta$.
Because $H$ and $G$ are additive across data points, any subset of quotes can be
removed or updated by simple algebraic subtraction of their local contributions.
This additive structure underpins the exactness of the proposed forgetting
operator.
Therefore, we build upon this simplistic yet structural insightful  observation of additivity in designing the unlearning operators.

\subsection{Sharded recompute operator}
\label{subsec:sharded}

The idea of dividing a dataset into shards for efficient unlearning has appeared in the
machine-learning literature, notably in the ``SISA'' framework of \cite{bourtoule2021machine}, which trains isolated submodels that can be retrained independently upon deletion requests.
Our sharded design  which we now present is mathematically different; rather than training independent submodels, we partition the Gauss--Newton normal equations themselves into additive
shard contributions $(H_k,G_k)$, allowing exact recomputation of the global system after a shard-level deletion.

Let $x_i$ be features (e.g., moneyness, maturity etc) and $y_i$ be the observed price (or implied volatility); together constituting the set option quotes $(x_i, y_i)$ with  $i \in D$ (finite index set).
We define $m(x; \theta)$ as the parametric pricing map (e.g., Heston) with parameters $\theta \in \Theta \subset \mathbb{R}^p$ and Loss as $\ell(y, \hat{y})$, typically squared error on prices.
And calibration minimizes the empirical loss of the form
\begin{equation}
J(\theta; D) = \sum_{i \in D} \ell(y_i, m(x_i; \theta)).
\end{equation}
%
%
%
%
At a reference $\theta$, We use Gauss--Newton to solve 
\begin{equation}
\quad H(\theta) \Delta \theta = g(\theta),
\end{equation}
where $g(\theta) = \sum_i J_i(\theta)^\top r_i(\theta)$, $H(\theta) = \sum_i J_i(\theta)^\top J_i(\theta)$, $r_i(\theta) = y_i - m(x_i; \theta)$, $J_i(\theta) = \nabla_\theta m(x_i; \theta)$.
Given a trained model on $D$, and a subset $F \subset D$ to "forget", update $\theta$ so the new parameter equals (or closely matches) the parameter you would obtain by retraining on $D \setminus F$, without sweeping the entire $D$ again.
Sharded recomputation first partitions the data set into $K$ shards $D = \bigcup_{k=1}^{K} D_k$.
Sharding could be by time (e.g., month).
We remark an important suggestion on the shard sizes. 
Given the assumption that some parts of data will be unlearned, keeping shards moderately sized so that removing a subset touches few shards .
For any reference $\theta$, we define per-shard Gauss--Newton aggregates:
\begin{equation}
G_k(\theta) = \sum_{i \in D_k} J_i(\theta)^\top r_i(\theta), \quad H_k(\theta) = \sum_{i \in D_k} J_i(\theta)^\top J_i(\theta).
\end{equation}
Global aggregates are sums across shards:
\begin{equation}
G(\theta) = \sum_{k=1}^{K} G_k(\theta), \quad H(\theta) = \sum_{k=1}^{K} H_k(\theta).
\end{equation}
The methodology with reference $\theta^{\text{ref}}$. 
First compute  $G_{k}(\theta^{\text{ref}})$, $H_{k}(\theta^{\text{ref}})$ for each shard $k$, sum to $G$, $H$; solve $H \Delta \theta = G$ and finally update $\theta \leftarrow \theta^{\text{ref}} + \Delta \theta$.
One can optionally relinearize once or twice (update reference and re-compute shard stats).
This, rather, describes the baseline calibration procedure when no data has been forgotten.
For $F \subset D$, unlearning phase includes identifying the set of affected shards $\mathcal{K}(F) = \{k : D_{k} \cap F \neq \emptyset\}$ first, then recomputing only those shards such that 
$k \in \mathcal{K}(F)$ on $D_{k} \setminus F$ to obtain $G_{k}^{\prime}$, $H_{k}^{\prime}$.
We note that unaffected shards keep their statistics $G_{k}$, $H_{k}$ while new global stats become
\begin{equation}\label{eq:G_prime_first}
G^{\prime} = \sum_{k \notin \mathcal{K}(F)} G_{k} + \sum_{k \in \mathcal{K}(F)} G_{k}^{\prime},
\end{equation}
\begin{equation}\label{eq:H_prime_first}
H^{\prime} = \sum_{k \notin \mathcal{K}(F)} H_{k} + \sum_{k \in \mathcal{K}(F)} H_{k}^{\prime},
\end{equation}
with \eqref{eq:G_prime_first} and \eqref{eq:H_prime_first}, we solve for $H^{\prime} \Delta \theta^{\prime} = G^{\prime}$ and update accordingly.
Although the discussion above is straightforward, presenting the intuition of sharding is timely.
The idea of grouping the option (panel) data into subgroups before calibration is already well known in the literature; for instance we refer to \cite{dumas1998implied}, \cite{ulrich2023implied}, \cite{homescu2011implied} and \cite{friedman2014some}.
As suggested, bucketing option data is already a common practice and we repurpose it in our analogy to share the same sharding narrative presented in \cite{bourtoule2021machine}.
On the broader terms, there are possible alternatives in sharding which will ultimately depend on the unlearning requests or necessities.
In our study, we simply do it by time as it is simple and stable.
This approach is particularly useful if the quotes arrive over days and forgetting targets a \emph{date range}.
Another possibility is through product structure, for instance one shard including ATM with fewer than 30 days to expiration and another for OTMs with fewer than 60 days to expiration.
Possible scenario might include forgetting particular short-dated options from the dataset or removing bad surface are around 1M tenor.
This could provide better locality than time-sharding and perhaps the one that could contribute to model validation procedures given sharding targets after specific questions raised.
Finally we could suggest a hybrid version of what we discussed so far yet we leave such possibilities for further studies.

To ensure validity, we have standard stability and regularity assumptions. 
We assume that $m(x; \theta)$ is twice continuously differentiable in $\theta$ on $\Theta$ so that the model's pricing map is smooth enough.
Further, the Jacobian $J(\theta)$ should exist and local Taylor expansion, $r(\theta + \Delta \theta) \approx r(\theta) + J \Delta \theta$, does not lead to instability in residuals.
Secondly we assume that per-shard sums $H_k(\theta^{\text{ref}})$ are positive semidefinite so that local curvature remains nonnegative; i.e. shards contribute non-negative information.
And lastly we assume that global $H(\theta^{\text{ref}})$ is positive definite with minimal eigenvalue $\lambda_{\min} > 0$ so that   local strong convexity is established at $\theta^{ref}$ and the global normal equation has a unique solution.

\begin{proposition}{(Shard-level exactness at a fixed linearization)}\label{prop:1}
Fix the reference $\theta^{\text{ref}}$. Consider the Gauss--Newton normal equations at that reference. If we recompute exactly $G'_k$, $H'_k$ for all affected shards on $D_k \setminus F$ and keep $(G_k, H_k)$ for unaffected shards, then the global system
$$\left(\sum_{k \notin \mathcal{K}(F)} H_k + \sum_{k \in \mathcal{K}(F)} H'_k\right) \Delta \theta' = \sum_{k \notin \mathcal{K}(F)} G_k + \sum_{k \in \mathcal{K}(F)} G'_k$$
is identical to the Gauss--Newton system built by running over the full retained set $D \setminus F$ at $\theta^{\text{ref}}$. Consequently, the update $\theta' = \theta^{\text{ref}} + \Delta \theta'$ matches full retraining under the same linearization.
\end{proposition}
Proof is trivial as both sides are linear sums over $i \in D \setminus F$; sharding is just a partition.
The reason we put forward Proposition \ref{prop:1} is straightforward.
Linearization, in our context, refers to the first-order approximation of the residual between model-generated prices and observed quotes, i.e., the loss surface around the current parameter estimate.
With this, Proposition \ref{prop:1} simply formalizes that under a fixed linearization, sharding recomputation is \emph{exactly equivalent} to full recalibration on the retained data (the data left upon removal of some data). 
Formally, the unlearning operation becomes linear in the data as it is the bridge from nonlinear calibration to unlearning operator and justifies our sharded recomputation unlearning operator as an analytically consistent replacement for full retraining.
%
The following result adapts the classical local error bound of Gauss--Newton iterations to our sharded-unlearning framework.
Here, the bound quantifies the accuracy gained after one local relinearization on affected shards, rather than the asymptotic convergence of a full iterative scheme.

\begin{proposition}[Accuracy after one relinearization on affected shards]\label{thm:relinearization}
	Under the same smoothness and strong-convexity assumptions as before,
	and additionally assuming that the Jacobian $J(\theta)$ is Lipschitz
	continuous in a neighborhood of $\theta^{\text{ref}}$
	Let $\theta' = \theta^{\text{ref}} + \Delta \theta'$ be the parameter produced by the shard-level update of Proposition~\ref{prop:1} at a fixed reference $\theta^{\text{ref}}$ on the retained dataset $D \setminus F$.
	Let $\hat{\theta}$ denote the parameter obtained by performing one relinearization at $\theta'$ and resolving the Gauss--Newton system on $D \setminus F$.
	Then there exist constants $C_1, C_2 > 0$ depending on $L_J$, $R_{\max}$, and the conditioning of $H'(\theta^{\text{ref}})$ such that
	\begin{equation}\label{eq:relinearization-bound}
		\|\hat{\theta} - \theta'\|
		\;\le\;
		C_1 \,\|r(\theta^{\text{ref}})\|\, \|\Delta \theta'\|
		\;+\;
		C_2 \,\|\Delta \theta'\|^2.
	\end{equation}
\end{proposition}

\begin{remark}[Quadratic accuracy under small residuals]\label{cor:quadratic}
	If, in addition, the residual norm at the reference satisfies
	$\|r(\theta^{\text{ref}})\| \le c\, \|\Delta \theta'\|$ for some $c>0$,
	then inequality~\eqref{eq:relinearization-bound} reduces to
	\[
	\|\hat{\theta} - \theta'\| \;\le\; C\, \|\Delta \theta'\|^2,
	\quad
	C := C_1 c + C_2.
	\]
	Hence, a single relinearization on the affected shards yields a second-order accurate refinement of the fixed-linearization update.
\end{remark}

\begin{proof}[Proof of Proposition~\ref{thm:relinearization}]
	Write aggregated quantities on $D \setminus F$ as
	\[
	H(\theta) = J(\theta)^\top J(\theta),
	\qquad
	G(\theta) = J(\theta)^\top r(\theta).
	\]
	By the Lipschitz property and twice differentiability of $m$, for $\Delta = \Delta \theta'$ we have
	\[
	\begin{aligned}
		J(\theta') &= J(\theta^{\text{ref}}) + E_J,
		\qquad &&\|E_J\| \le L_J \|\Delta\|,\\
		r(\theta') &= r(\theta^{\text{ref}}) + J(\theta^{\text{ref}})\Delta + R_r,
		\qquad &&\|R_r\| \le C_r \|\Delta\|^2,
	\end{aligned}
	\]
	for some constant $C_r = O(L_J)$. 
	Expanding $G(\theta')$ and $H(\theta')$ gives
	\[
	\begin{aligned}
		G(\theta') &= J(\theta')^\top r(\theta')
		= J^\top r + J^\top J \Delta
		+ E_J^\top r + O(\|\Delta\|^2),
		\\
		H(\theta') &= J(\theta')^\top J(\theta')
		= H(\theta^{\text{ref}}) + \Delta H,
		\qquad \|\Delta H\| \le C_H \|\Delta\|,
	\end{aligned}
	\]
	with $C_H = O(\|J(\theta^{\text{ref}})\| L_J + L_J^2)$.
	Let $\Delta^+$ solve the relinearized system $H(\theta')\Delta^+ = G(\theta')$.
	Subtracting the fixed-linearization equation $H(\theta^{\text{ref}})\Delta = G(\theta^{\text{ref}})$ yields
	\[
	H(\theta')(\Delta^+ - \Delta)
	= E_J^\top r + O(\|\Delta\|^2) + \Delta H \, \Delta.
	\]
	Taking norms and using $\|H(\theta')^{-1}\| \le 2/\lambda_{\min}$ for $\|\Delta\|$ small gives
	\[
	\|\Delta^+ - \Delta\|
	\;\le\;
	\frac{2L_J}{\lambda_{\min}}\, \|r(\theta^{\text{ref}})\|\, \|\Delta\|
	\;+\;
	C_2' \|\Delta\|^2,
	\]
	for a constant $C_2'$ depending on $L_J$ and $\lambda_{\min}^{-1}$.  Since
	$\hat{\theta}-\theta'=(\theta^{\text{ref}}+\Delta^+)-(\theta^{\text{ref}}+\Delta)=\Delta^+ - \Delta$,
	this proves~\eqref{eq:relinearization-bound} with $C_1 = 2L_J / \lambda_{\min}$ and $C_2 = C_2'$.
	The corollary follows by substituting $\|r(\theta^{\text{ref}})\|\le c\|\Delta\|$ and absorbing constants.  \qedhere
\end{proof}

The quadratic accuracy bound derived above  is structurally related to the classical local error analysis of Gauss--Newton and Newton--Kantorovich iterations. 
Here, however, the theorem is not invoked to study asymptotic convergence of an iterative solver, but to establish the \emph{fidelity of 	machine unlearning} within a sharded calibration framework. 
In our setting, the relinearization step is applied only to the affected shards after a data-deletion event, and the resulting bound quantifies how closely this partial update reproduces the fully retrained Gauss--Newton solution on the retained data set.
The adaptation of a classical local error argument to the context of selective unlearning therefore provides new insight into the stability and precision of unlearning operations in financial model
calibration.
The following   adapts a standard perturbation bound for linear systems   to our Gauss--Newton unlearning update. 
It provides an upper limit on the parameter deviation induced by downdating the curvature and gradient terms.

\begin{proposition}[Stability of the unlearning update]\label{thm:stability}
	Let $\theta=\theta^{\mathrm{ref}}+\Delta\theta$ solve the fixed linearization system
	$H\Delta\theta=G$ on $D$, and let $\theta'=\theta^{\mathrm{ref}}+\Delta\theta'$ solve
	$H'\Delta\theta'=G'$ on $D\setminus F$, where
	\[
	H' = H + \Delta H,\qquad G' = G + \Delta G.
	\]
	Assume $H'$ is invertible. Then
	\begin{equation}\label{eq:stability}
		\|\theta'-\theta\| \;=\; \|\Delta\theta'-\Delta\theta\|
		\;\le\;
		\|H'^{-1}\| \,\big(\,\|\Delta G\| + \|\Delta H\|\,\|\Delta\theta\|\,\big).
	\end{equation}
	Moreover, if $\|H^{-1}\|\,\|\Delta H\|<1$, then
	\[
	\|H'^{-1}\| \;\le\; \frac{\|H^{-1}\|}{1-\|H^{-1}\|\,\|\Delta H\|}.
	\]
\end{proposition}

\begin{proof}
	From $H\Delta\theta=G$ and $(H+\Delta H)\Delta\theta' = G+\Delta G$,
	\[
	H'(\Delta\theta'-\Delta\theta) \;=\; \Delta G - \Delta H\,\Delta\theta.
	\]
	Multiply by $H'^{-1}$ and take norms; the Neumann bound follows from
	$H'^{-1} = (I+H^{-1}\Delta H)^{-1}H^{-1}$ whenever $\|H^{-1}\Delta H\|<1$.
\end{proof}

\begin{remark}[Robust loss control]\label{cor:robust}
	Suppose the per-quote loss is Huber with threshold $c>0$ and residuals are
	locally bounded. Then each quote's influence function is bounded by $c$,
	so there exist constants $C_J,C_{JJ}$ (depending on Jacobian norms) such that,
	for forgetting $F$,
	\[
	\|\Delta G\| \;\le\; C_J\,c\,|F|, \qquad
	\|\Delta H\| \;\le\; C_{JJ}\,|F|.
	\]
	Consequently, if $H'$ is well conditioned,
	\[
	\|\theta'-\theta\| \;\lesssim\; \kappa(H')\big(c\,|F| + |F|\,\|\Delta\theta\|\big),
	\]
	i.e., the change scales linearly with the forgotten mass and the conditioning.
\end{remark}

\begin{remark}[Conditioning links: eigenvalues and Neumann bound]
	Assume throughout the spectral (2-)norm and that $H,H'$ are symmetric positive definite.
	Then $\|H'^{-1}\|_2=1/\lambda_{\min}(H')$, and the stability estimate
	\[
	\|\Delta\theta'-\Delta\theta\|
	\;\le\; \frac{\|\Delta G\|_2 + \|\Delta H\|_2\,\|\Delta\theta\|_2}{\lambda_{\min}(H')}
	\]
	shows that larger $\lambda_{\min}(H')$ (better conditioning) improves robustness.
	Moreover, if $\|H^{-1}\Delta H\|_2<1$, the Neumann expansion yields
	\[
	\|H'^{-1}\|_2 \;\le\; \frac{\|H^{-1}\|_2}{\,1-\|H^{-1}\Delta H\|_2\,}
	\;=\; \frac{1}{\,\lambda_{\min}(H)\,\big(1-\|H^{-1}\Delta H\|_2\big)}\,,
	\]
	and consequently
	\[
	\|\Delta\theta'-\Delta\theta\|
	\;\le\;
	\frac{\|\Delta G\|_2 + \|\Delta H\|_2\,\|\Delta\theta\|_2}
	{\lambda_{\min}(H)\,\big(1-\|H^{-1}\Delta H\|_2\big)}.
	\]
	Finally, by Weyl's inequality,
	$\lambda_{\min}(H') \ge \lambda_{\min}(H) - \|\Delta H\|_2$,
	so $H'$ remains positive definite whenever $\|\Delta H\|_2<\lambda_{\min}(H)$.
\end{remark}

The above results collectively ensure that the curvature downdate remains numerically stable and the Gauss--Newton step is well defined under moderate forgetting.
Having established the local stability and conditioning properties, we next turn to the fast refactorization approach.

\subsection{Fast refactor operator}
\label{sec:fast-refactor}

Machine unlearning is not defined by the speed of recomputation, but rather by the semantics of the data removal.
We could formalize it such that an algorithm carries an unlearning spirit after deleting subset $F \subset D$ and the resulting model parameters are \emph{indistinguishable} from those obtained by retraining on $D \setminus F$.
Therefore, in the sense of \cite{bourtoule2021machine}, the sharded recomputation remains a legitimate unlearning operator, laying the conceptual definition of unlearning in our framework.
Given our points of concern embark on computational capability rather than on issues related to the well articulated purposes of machine unlearning (eg. privacy), we remark the necessity of offering an efficient implementation of that same operator for reasons we discuss shortly. 

We now introduce a faster data-free forgetting operator that yields the same Gauss--Newton (GN) update as retraining on the retained set, without accessing raw quotes once a cache is built.
Throughout, we fix a reference parameter $\theta^{\mathrm{ref}} \in \Theta$ and work with the GN normal equations at this reference.
Let $F\subset D$ denote a subset of quotes to be forgotten and $(H',G')$ denote the post-forgetting aggregates obtained by subtracting the contributions of $F\subset D$.  
Given the Gauss--Newton aggregates are linear in $\{u_i,\psi_i\}$, the effect of removing
$F$ can be represented exactly by subtraction:
\begin{equation}
H' \;=\; H - \sum_{i\in F} \psi_i, \qquad
G' \;=\; G - \sum_{i\in F} u_i.
\end{equation}
The updated parameter is then obtained by solving once
\begin{equation}
	\label{eq:fast-operator}
	(H' + \lambda I)\,\Delta\theta' = G', \qquad
	\theta' = \theta^{\mathrm{ref}} + \Delta\theta'.
\end{equation}
Equations \eqref{eq:fast-operator} require only the cached statistics $(u_i,\psi_i)$, not the raw market data $(x_i,y_i)$, and thus implement a \emph{data-free forgetting operator}. 
At the fixed linearization $\theta^{\mathrm{ref}}$, this refactoring exactly removes the influence of the
forgotten subset from the calibration system. 
The procedure achieves the same solution as a full retraining on $D\setminus F$, up to machine precision, while
avoiding all repricing and re-differentiation.
Fast refactorization operates under the same local regularity conditions	introduced in Section~\ref{subsec:sharded}, namely smoothness, 	local strong convexity (possibly enforced via a small ridge term 	$\lambda I$), Lipschitz continuity of the Jacobian, and bounded residuals.
These ensure that the refactorized system $(H' + \lambda I)\Delta\theta' = G'$ 	remains well-posed and that all prior analytical results remain valid.

Each option quote $i$ contributes via its residual
$r_i(\theta^{\mathrm{ref}})=y_i-m(x_i;\theta^{\mathrm{ref}})$ and local sensitivity $J_i(\theta^{\mathrm{ref}})=\nabla_\theta m(x_i;\theta^{\mathrm{ref}})$.
Define $u_i := J_i^\top r_i$ and $\psi_i := J_i^\top J_i$. 
Then $G=\sum_i u_i$ and $H=\sum_i \psi_i$ are the Gauss--Newton aggregates at $\theta^{\mathrm{ref}}$, and the update solves $H\,\Delta\theta=G$.
Thus the collection $\{(u_i,\psi_i)\}_i$ is \emph{algebraically sufficient for the 	linearized calibration at $\theta^{\mathrm{ref}}$}: once stored, the influence of any subset $F$ can be removed exactly by subtraction, $H' = H - \sum_{i\in F}\psi_i$ and $G' = G - \sum_{i\in F}u_i$,
without revisiting raw data. 
If a robust loss is used, the same identities hold with per-quote weights ($u_i = w_i J_i^\top r_i$, $\psi_i = w_i J_i^\top J_i$).
These statistics are tied to the chosen reference; upon relinearization ($\theta^{\mathrm{ref}}\mapsto \theta^{\mathrm{new}}$), the pairs $(u_i,\psi_i)$ should be recomputed at the new reference.

Although we build fast refactor operator on the foundations of the sharded recomputation, the latter does not actually need shards.
During the initial training we build caches, $(H, G)$, and per-shard aggregates, $(H_k, G_k)$.
In the sharded recomputation, with some quotes removed, we reopen only the shards that contain them, then recompute $(H_k, G_k)$ for those shards and sum up with others.
In fast refactor we go one step further and directly subtract each forgotten quote's contribution from the cached global $(H, K)$ so that there no longer exists the need to reopen or recompute the shards, making it completely data-free and instantaneous.
Inclusion of shards, then, in fast refactor might seem contradictory.
However, we remark that the shards play important roles in categorizing the forgetting set (although removal is not driven by shards) and more importantly it provides security in the case some of quotes lacking cached Jacobians, $J_i$.

Technically speaking, in the fast refactor variant, unlearning operates at quote granularity as once per-quote GN statistics $(J_i^T J_i, J_i^T r_i) $ are cashed at the reference point, removing any subset $F \subset D$ amounts to simplistic algebraic operations of the global normal equations, unrelated to how the data were originally sharded. 
In our implementation, the fast refactorization step forms $H' = H - \sum_{i\in F} J_i^\top J_i$ and $G' = G - \sum_{i\in F} J_i^\top r_i$ explicitly, followed by a fresh Cholesky factorization of $H'$.
This retains exactness under the fixed linearization while avoiding any recomputation over the retained dataset.
Although a true rank-1 Cholesky downdate would further reduce cost to $O(p^2 |F|)$, we found the explicit rebuild to be numerically safer and sufficiently fast for moderate $p$.

While the concept of subtracting per-sample statistics is trivial for linear models,
it becomes nontrivial for nonlinear calibration because the residuals and Jacobians
depend on the current parameter estimate. 
Naively removing quotes invalidates the current linearization, so retraining from scratch remains the default.
Yet, in practice, calibration pipelines already store large intermediate structures' 
per-quote sensitivities, residuals, and curvature estimates for diagnostic or
parallel-computation purposes.
This suggests the possibility of an \emph{operator-level unlearning} mechanism: removing data by algebraic downdating of the cached normal equations, without reprocessing the raw option surface.
In our framework, this takes the form of Gauss--Newton updates on refactored $(H',G')$, achieving exact unlearning at a fixed linearization.

The subtraction step removes the statistical influence of each forgotten quote because,
under Gauss--Newton linearization, the normal equations decompose additively across data points.
Each quote $i$ contributes $(\psi_i, u_i) = (J_i^\top J_i,\, J_i^\top r_i)$ to the global
system $(H, G)$. Solving $(H+\lambda I)\Delta\theta = G$ thus depends on the data only
through these linear aggregates. Forgetting a subset $F$ corresponds to replacing
\begin{equation}
(H', G') = \big(H - \sum_{i\in F}\psi_i,\; G - \sum_{i\in F}u_i\big),
\end{equation}
which is identical to the system built on $D\setminus F$.
Consequently, the updated parameter $\theta' = \theta^{\mathrm{ref}} + \Delta\theta'$
matches the retraining result under the same linearization, with no residual dependence on $F$.
This equality $H'=H^\star$, $G'=G^\star$ is formalized in such as:


\begin{proposition}[Exactness under fixed linearization]\label{prop:exact2}
	Let $H^\star, G^\star$ denote the Gauss--Newton aggregates constructed directly on
	the retained set $D\setminus F$ at $\theta^{\mathrm{ref}}$:
	\[
	H^\star = \sum_{i\in D\setminus F} J_i(\theta^{\mathrm{ref}})^\top J_i(\theta^{\mathrm{ref}}),
	\qquad
	G^\star = \sum_{i\in D\setminus F} J_i(\theta^{\mathrm{ref}})^\top r_i(\theta^{\mathrm{ref}}).
	\]
	Then $H^\star=H'$ and $G^\star=G'$, where
	$(H',G')=(H-\sum_{i\in F}\psi_i,\, G-\sum_{i\in F}u_i)$ are the refactored
	aggregates at $\theta^{\mathrm{ref}}$.
	Consequently, for the same $\lambda\ge 0$, the update $\theta'$ produced by
	\((H'+\lambda I)\Delta\theta'=G'\) coincides with the parameter obtained by
	retraining the Gauss--Newton system on $D\setminus F$ at $\theta^{\mathrm{ref}}$.
\end{proposition}

\begin{proof}
	By the additive decompositions at $\theta^{\mathrm{ref}}$,
	$H=\sum_{i\in D} \psi_i=\sum_{i\in D} J_i^\top J_i$ and
	$G=\sum_{i\in D} u_i=\sum_{i\in D} J_i^\top r_i$.
	Subtracting forgotten contributions gives
	$H' = H - \sum_{i\in F}\psi_i = \sum_{i\in D\setminus F} J_i^\top J_i = H^\star$
	and similarly $G' = \sum_{i\in D\setminus F} J_i^\top r_i = G^\star$.
	Thus the regularized systems $(H'+\lambda I)\Delta\theta'=G'$ and
	$(H^\star+\lambda I)\Delta\theta^\star=G^\star$ are identical, yielding
	$\Delta\theta'=\Delta\theta^\star$ and hence the same $\theta'$.
\end{proof}

\begin{remark}

The operator we suggest in this part acts solely on the cached per-quote statistics $(u_i,\psi_i)$ and the precomputed global aggregates $(H,G)$ at $\theta^{\mathrm{ref}}$. 
It therefore removes the influence of the forgotten set $F$ exactly under the Gauss--Newton linearization without accessing any raw market quotes or re-evaluating model prices.

\end{remark}

Let $\Delta H := H'-H$ and $\Delta G := G'-G$. Let $\Delta\theta$ and $\Delta\theta'$ be the GN steps at $\theta^{\mathrm{ref}}$ on $D$ and $D\setminus F$, respectively, both with the same $\lambda$.
We use the vector $2-$norm and the induced operator norm for matrices.

\begin{proposition}[Linearized stability]\label{prop:stability2}
	 With the same $\lambda$ (so $H'+\lambda I$ is invertible),
	\[
	\|\Delta\theta' - \Delta\theta\|
	\;\le\; \| (H'+\lambda I)^{-1}\|\;\Big(\|\Delta G\| + \|\Delta H\|\,\|\Delta\theta\|\Big).
	\]
	In particular, if $|F|/|D|$ is small and $H'+\lambda I$ is well conditioned, then $\|\Delta\theta' - \Delta\theta\|$ is small.
\end{proposition}

\begin{proof}[Sketch]
	Write $(H+\lambda I)\Delta\theta=G$ and $(H'+\lambda I)\Delta\theta'=G'$. Subtract to obtain
	\[
	(H'+\lambda I)(\Delta\theta'-\Delta\theta) \;=\; \Delta G - \Delta H\,\Delta\theta,
	\]
	then multiply by $(H'+\lambda I)^{-1}$ and take norms.
\end{proof}

\begin{proposition}[Accuracy after one relinearization]\label{prop:relin2}
	Let $\theta^\star$ denote the (local) least-squares solution on $D\setminus\mathcal F$.
	Assume: (i) $J(\theta)$ is Lipschitz in a neighborhood of $\theta^\star$ with constant $L_J$,
	(ii) $J(\theta^\star)$ has full column rank, and
	(iii) the residual at the solution is small, $\|r(\theta^\star)\|\le \varepsilon$.
	Let $\widehat{\theta}$ be the GN/LM (Levenberg--Marquardt) solution obtained on $D\setminus\mathcal F$ after
	one relinearization at $\theta^{\mathrm{ref}}+\Delta\theta'$ (same $\lambda$).
	Then there exist constants $C_1,C_2>0$ (depending on $L_J$, local bounds, and $\|(H'+\lambda I)^{-1}\|$) such that
	\[
	\|\widehat{\theta}-\theta^\star\|
	\;\le\; C_1\,\|\Delta\theta'\|^2 \;+\; C_2\,\varepsilon\,\|\Delta\theta'\|.
	\]
	In particular, in the small-residual regime $(\varepsilon\approx 0)$,
	$\|\widehat{\theta}-\theta^\star\| = \mathcal O(\|\Delta\theta'\|^2)$.
\end{proposition}

\begin{proof}[Sketch]
	Standard Gauss--Newton local analysis (Newton--Kantorovich style):
	the model error from relinearization is $\mathcal O(\|\Delta\theta'\|^2)$ by Lipschitz $J$.
	Mapping this through the normal equations introduces $\|(H'+\lambda I)^{-1}\|$.
	The residual term yields the mixed $\varepsilon\,\|\Delta\theta'\|$ contribution.
\end{proof}

Shard-level recompute is exact at a fixed linearization by linearity of sums yet it still requires opening the affected shards.
The refactor operator strengthens this to a \emph{data-free} update by using cached per-quote statistics.
In linearity, forgetting is exact by subtraction 
Our result lifts this idea to \emph{nonlinear} parametric models via Gauss--Newton linearization, providing (to our knowledge) the first data-free exact forgetting operator for nonlinear least squares in financial calibration.
Once $(u_i,\psi_i)$ are retained and raw quotes purged, subsequent unlearning requests are executed algebraically. 
If $(u_i,\psi_i)$ are deemed sensitive, they can be encrypted or perturbed; bounds in Proposition~\ref{prop:stability2} quantify the resulting parameter drift.
%
%
We remark that the cache does not contain raw market quotes or strike--maturity grids.
Instead it stores derived Jacobian vectors $u_i=J_i$ and $ \psi_i=J_i r_i$, together with the global normal equations $H=\sum_i J_i^\top J_i$, $G=\sum_i J_i r_i$
evaluated at the calibrated parameter $\theta^\star$.
These quantities are sufficient for a local Gauss--Newton update but
contain no reconstructive information about individual data points.

\begin{remark}
Machine unlearning is effected by removing the contributions of
forgotten quotes from the cached statistics:
\[
H' = H - \sum_{i\in\mathcal{F}} J_i^\top J_i, \qquad
G' = G - \sum_{i\in\mathcal{F}} J_i r_i.
\]
The updated parameters
\(\theta_{\text{fast}}=(H'+\lambda I)^{-1}G'\)
coincide with those obtained by full retraining on the retained set
(up to numerical precision).
Hence the method satisfies the formal definition of \emph{machine unlearning}
as the ability to expunge a subset's influence from the trained model
without re-accessing the original data, \cite{guo2019certified, bourtoule2021machine}.
\end{remark}

We have several remarks that we believe is timely. 
While the proposed unlearning operators reproduce the Gauss--Newton update on the retained dataset, occasional increases in performance metrics (RMSE, in our case) may still be observed when unlearning is applied to very small subsets of option quotes. 
Two conceptually distinct mechanisms may explain this behavior.
A single Gauss--Newton step is locally exact only in a neighborhood of a previously converged solution. 
When the forget set is small, the displacement of the optimum is also small,  and one correction typically recovers the new minimizer to machine precision. 
However, when the retained dataset becomes extremely small, the linearization may no longer be valid, and the residuals may appear unstable.
%
Another one is, independently of the Gauss--Newton linearization error, numerical instability 
could arise if the curvature matrix $H'$ becomes ill-conditioned after forgetting\footnote{This could especially amplify if you remove a few influential quotes (e.g. some maturities or deep-OTMs that strongly shape volatility) which leaves the new curvature matrix with extremely small eigenvalues (large condition number) and perhaps with directions in parameter space that are almost unconstrained by the retained data.}. 
This occurs when the retained quotes provide insufficient informational content, leading to exceedingly small eigenvalues and nearly unconstrained parameter directions. 
In such cases, even correct cached gradients could produce large parameter  excursions and elevated RMSE. 
Such optimization refinements fall outside the scope of our unlearning study as it is sampling related rather than the methodology.
In summary, the occasional RMSE deviations observed for extremely small retained 
datasets arise from the standard numerical behavior of Gauss--Newton and Levenberg--Marquardt schemes, rather than from the unlearning operators themselves, \cite{ait1998nonparametric}.

Before closing the section, we include another proposition on the computational complexity of the calibration and both unlearning operators. 
We defer this result to the end of the section so as not to interrupt the flow of the theoretical development in the previous subsections as the developed approaches are not related to the Heston model that we use for exemplary purposes. 
Given  our framework is fully model-agnostic and applies to any differentiable option pricing map, the complexity statement becomes most transparent when expressed for the Heston model, where each price evaluation is performed via a Fourier--Simpson integral with $N_u$ frequency nodes. 
The following proposition therefore specializes the analysis to this setting, which is also the one used in our numerical experiments.

\begin{proposition}[Computational complexity of calibration and unlearning operators]
	\label{prop:complexity}
	Let $N$ denote the number of option quotes, $N_u$ the number of Fourier--Simpson integration nodes used in the Heston pricer, and $p$ the dimension of the parameter vector (e.g.\ $p=5$ for the Heston model). Assume $p$ is fixed and small. Then, under a Gauss--Newton calibration scheme at a fixed reference $\theta^{\mathrm{ref}}$, the following complexity bounds hold:
	\begin{enumerate}[(i)]
		\item \textbf{Full recalibration.} A single Gauss--Newton iteration on a dataset of size $N$ has time complexity
		\[
		T_{\mathrm{retrain}} = \mathcal{O}(N\,N_u) + \mathcal{O}(p^3),
		\]
		where the dominant cost is the evaluation of $N$ Heston prices via Fourier--Simpson quadrature. The $\mathcal{O}(p^3)$ term arises from assembling and solving the $p\times p$ normal equations.
		
		\item \textbf{Sharded recomputation.} Let $D = \bigcup_{k=1}^K D_k$ be a partition of the data into $K$ shards and let $K(F) \subseteq \{1,\dots,K\}$ denote the set of shards affected by a forget set $F \subset D$. Denote by $N_{\mathrm{eff}}(F)$ the number of quotes in $\bigcup_{k\in K(F)} D_k$. Then a sharded recomputation step has time complexity
		\[
		T_{\mathrm{recomp}}(F) = \mathcal{O}\big(N_{\mathrm{eff}}(F)\,N_u\big) + \mathcal{O}(p^3),
		\]
		i.e.\ it is equivalent to a full Gauss--Newton step restricted to the affected shards. In the worst case $N_{\mathrm{eff}}(F) \approx N$, and $T_{\mathrm{recomp}}(F)$ degenerates to $T_{\mathrm{retrain}}$.
		
		\item \textbf{Fast refactorization.} Suppose that, at $\theta^{\mathrm{ref}}$, per-quote Gauss--Newton statistics
		\[
		u_i = J_i(\theta^{\mathrm{ref}})^\top r_i(\theta^{\mathrm{ref}}), 
		\qquad
		\psi_i = J_i(\theta^{\mathrm{ref}})^\top J_i(\theta^{\mathrm{ref}})
		\]
		and the global aggregates $H = \sum_{i} \psi_i$, $G = \sum_{i} u_i$ have been cached. Then a fast refactorization unlearning request for a forget set $F$ can be carried out in
		\[
		T_{\mathrm{fast}}(F) = \mathcal{O}(|F|\,p^2) + \mathcal{O}(p^3)
		\]
		time, corresponding to subtracting $\{\psi_i,u_i\}_{i\in F}$ from $(H,G)$ and solving the refactored $p\times p$ linear system. For fixed $p$, this is $\mathcal{O}(|F|) + \mathcal{O}(1)$, independent of $N$ and $N_u$.
	\end{enumerate}
	In particular, for fixed $p$, both full recalibration and sharded recomputation scale linearly in $N$ and approximately linearly in $N_u$, i.e.\ $\mathcal{O}(N\,N_u)$ in the dominant term, whereas the fast refactorization operator has per-request complexity independent of $N$ and $N_u$ and grows only with the size of the forget set $F$.
\end{proposition}

\begin{proof}
	For each quote $i$, evaluation of the Heston price $m(x_i;\theta)$ via Fourier--Simpson quadrature requires $\mathcal{O}(N_u)$ floating point operations, as the characteristic function $\varphi(u;\theta)$ and the integrand are evaluated at $N_u$ frequency nodes and combined by a weighted summation. Thus, pricing all $N$ quotes at a fixed parameter vector costs $\mathcal{O}(N\,N_u)$ operations. The Gauss--Newton step additionally forms residuals $r_i$ and Jacobians $J_i$, and accumulates
	\[
	H = \sum_{i=1}^N J_i^\top J_i, 
	\qquad
	G = \sum_{i=1}^N J_i^\top r_i,
	\]
	which require at most a constant factor overhead per quote when $p$ is fixed. Solving the normal equations $(H+\lambda I)\Delta\theta = G$ by, e.g., Cholesky factorization, has cost $\mathcal{O}(p^3)$. This proves (i).
	
	For sharded recomputation, only quotes in shards $k\in K(F)$ are repriced and their Jacobians recomputed, while unaffected shards reuse their cached $(H_k,G_k)$. If $N_{\mathrm{eff}}(F)$ denotes the total number of quotes in the affected shards, then the cost of recomputing their contributions is $\mathcal{O}(N_{\mathrm{eff}}(F)\,N_u)$, followed by the same $\mathcal{O}(p^3)$ solve on the updated global system. In the worst case, if $K(F)=\{1,\dots,K\}$, then $N_{\mathrm{eff}}(F)\approx N$ and the complexity coincides with full recalibration, establishing (ii).
	
	For fast refactorization, no repricing or Jacobian evaluation is performed once the cache is built. Each forgotten quote $i\in F$ contributes a rank-one downdate\footnote{Each forgotten quote $i$ contributes a small $p\times p$ matrix 
		$\psi_i = J_i^\top J_i$ and vector $u_i = J_i^\top r_i$ to the downdate 
		of $(H,G)$.
	} $\psi_i \in \mathbb{R}^{p\times p}$ and a vector downdate $u_i \in \mathbb{R}^p$ to $(H,G)$. Updating
	\[
	H' = H - \sum_{i\in F} \psi_i, 
	\qquad
	G' = G - \sum_{i\in F} u_i
	\]
	requires $\mathcal{O}(|F|\,p^2)$ operations. A fresh Cholesky factorization of $H'$ and back-substitution then cost $\mathcal{O}(p^3)$. Since $p$ is fixed and small, these costs are independent of $N$ and $N_u$, and the total complexity is $\mathcal{O}(|F|\,p^2) + \mathcal{O}(p^3)$, proving (iii).
\end{proof}

We finally note that our exactness statements are always with respect to the Gauss--Newton linearization at a reference point, not claiming global equivalence of fully iterated nonlinear solvers. 
In practice, however, calibration is often near an optimum and a single GN/LM step is used as a local adjustment, which is precisely the regime our operators are designed for.

\section{Illustrations}

First\footnote{All scripts are in Python and all associated numerical illustrations presented in this manuscript are	carried out on a system with i7 Core with 2.20 GHz and 16 GB RAM.}, we remark that in Figures of this section, unless stated otherwise, report the median over 3-5 random unlearning realizations per fraction, ensuring robustness to randomness in forgotten subsets.
We synthetically generate a surface of European call prices under the Heston model with known ground-truth parameters such  as $ \theta_{\mathrm{true}} = (\kappa,\theta_v,\sigma_v,\rho,v_0)= (2.0,\,0.06,\,0.30,\,-0.6,\,0.06), $ with $ r=0.01, S_0=100.$
We generate a path of either 90 trading days for a small sample experiment  or 180 trading days for a large sample experiment (Euler--Maruyama with correlated Brownian shocks; with $\Delta t=1/252$ throughout this section).
For each day in the path we form European call quotes at maturities  $T\in\{30,60\}$ days (i.e., $\{30,60\}/252$ years) and strikes  \(X\in\{90,100,110\}\) for small sample experiment and  maturities $T\in\{30,60,90\}$ days and strikes $X\in\{80,90,100,110,120\}$ for large sample experiment. 
Option prices \(m(x_i;\theta)\) are computed with the semi-analytic Heston formula  via Fourier inversion and Simpson's rule with either \(U_{\max}=50\) and \(N_u=180\) nodes or \(U_{\max}=120,\;N_u=800\) again depending on the sample size.
To emulate measurement noise, we perturb each price either by \(\varepsilon_i\sim\mathcal{N}(0,\sigma^2)\) with \(\sigma=10^{-3}\) or  \(\varepsilon_i\sim\mathcal{N}(0,(5\times10^{-4})^2)\) for the small and large sample size, respectively.
We remark that the total number of quotes, $N$, depend  on the path horizon and coverage and is thus not fixed a priori.  
We partition the quotes by calendar time into contiguous shards of  either 10 or 30 days, similarly.
 %
%
Starting from  $\theta^{\text{ref}}=(1.0,\,0.04,\,0.20,\,-0.3,\,0.04)$, we run a short Levenberg--Marquardt loop (Gauss--Newton with adaptive damping)  to obtain $\theta^\star$.
At \(\theta^\star\), we compute central finite-difference Jacobians $J_i=\nabla_\theta m(x_i;\theta^\star)$ and residuals \(r_i=y_i-m(x_i;\theta^\star)\).
We cache, for each quote $i$, $u_i, \psi_i \in \mathbb{R}^5 $ and the global normal equations, $(H, G)$ together with a small Tikhonov term \(\lambda=10^{-6}\) (i.e., we solve with \(H+\lambda I\)).
We also store per-shard, $(H_k, G_k)$ given a forget set \(\mathcal{F}\subset\{1,\dots,N\}\).

Speaking of the  cache, we remind that it is not a replica of the training data.
It stores only derivative-based sufficient statistics that summarize the model's local curvature at the calibration optimum.
Forgetting operates by algebraic removal of those statistics associated with the forgotten samples, which is exactly what the unlearning literature defines as \emph{data deletion at the parameter level}.
No raw strikes, maturities, or prices are revisited once the cache is built.

Our baseline, as we mentioned before, is the recalibration, full retraining in which we  recompute all Jacobians and residuals on the retained subset from the raw quotes, thereby rebuilding the normal equations $(H', G')$
and take one Gauss--Newton step to obtain \(\theta_{\text{ret}} = (\theta^{\text{ref}} + (H'+\lambda I)^{-1}G')\).
This represents a full recalibration from scratch and serves as the ground--truth
baseline for evaluating the two unlearning operators.

All timings are on a single laptop core (NumPy/BLAS pinned to one thread).
Because each calibration step involves only a few thousand Heston price evaluations and a single $5\times5$ linear solve, wall-clock runtimes are on the order of seconds even in the full configuration ($U_{\max}=120,\,N_u=800$).
Subsequent unlearning operations reuse cached Jacobians and require no re-pricing, yielding sub-second updates.
This behavior is consistent with the $\mathcal{O}(p^3)$ cost of the
Gauss--Newton linear system and the modest number of Fourier nodes per price evaluation.
More intuitively speaking, it is as we are not performing a full market-scale optimization but a one-step linearized Gauss--Newton update on a small synthetic grid.
We also stress that every plotted point is a median across independent random forget sets so that the figures already represents typical behavior, not a single lucky case.

We start by showing equivalence between retraining and our proposed machine unlearning approach, the fast factorization in this case.
By equivalence, we mean that up to machine precision (mostly on the order of $10^{-13}$); both approaches provide  identical results.
We observe an exemplary comparison in Figure~\ref{plot:precisions1} with the Heston variables separately shown, and the $y$-axis shows the difference of $\theta_{\text{fast}} -\theta_{\text{retrain}}$.
We remark again that we reserve $\theta_v$ for a parameter of the Heston model, and $\theta$ for the parameter space of the Heston model.
Subfigure~\ref{fig:prec1} shows per-parameter distributions of the same differences across all runs. The difference distribution of each parameter collapses around zero, and even that $\kappa$ has broader variance it still is on the order of $10^{-13}$; meaning purely numerical floating-point variation.   
Another note is that the \emph{vase} in $\kappa$ shows slightly more spread around zero, but still zero bias.

\begin{figure}[H]
	\begin{subfigure}{0.5\textwidth}
		\centering
		\includegraphics[ width=1\linewidth]{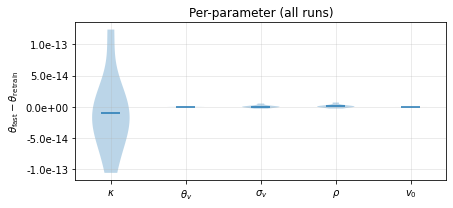}
		\caption{Per-parameter (all runs) }
		\label{fig:prec1}
	\end{subfigure}
	\begin{subfigure}{0.5\textwidth}
		\centering
		\includegraphics[width=1\linewidth]{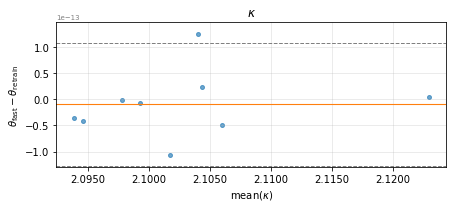}
		\caption{$\kappa$ }
		\label{fig:prec2}
	\end{subfigure}
	\begin{subfigure}{0.5\textwidth}
		\centering
		\includegraphics[width=1\linewidth]{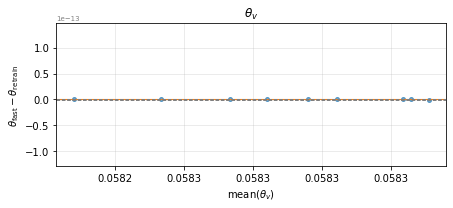}
		\caption{ $\theta_v$}
		\label{fig:prec3}
	\end{subfigure}
	\begin{subfigure}{0.5\textwidth}
		\centering
		\includegraphics[width=1\linewidth]{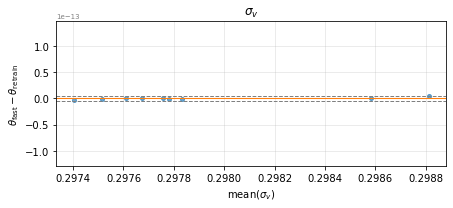}
		\caption{ $\sigma_v$}
		\label{fig:prec4}
	\end{subfigure}
	\begin{subfigure}{0.5\textwidth}
		\centering
		\includegraphics[width=1\linewidth]{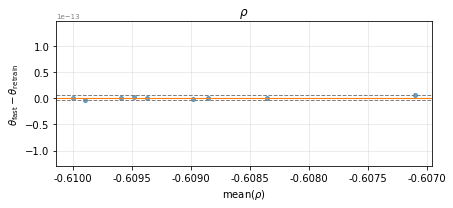}
		\caption{ $\rho$}
		\label{fig:prec5}
	\end{subfigure}
	\begin{subfigure}{0.5\textwidth}
		\centering
		\includegraphics[width=1\linewidth]{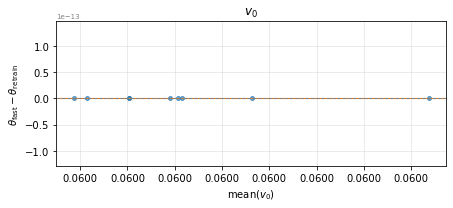}
		\caption{ $v_0$ }
		\label{fig:prec6}
	\end{subfigure}
	\caption{An exemplary comparison of equivalence of retraining and fast factorization with a smaller sample  }
	\label{plot:precisions1}
\end{figure}

Subfigures~\ref{fig:prec2}--\ref{fig:prec6} compare parameter wise, while $x$-axis indicates the mean of the two estimates (of retraining and fast factorization) for each parameter across runs (i.e., trials or experiments). 
We remark that the vertical scale is around $10^{-13}$, and expectedly $\kappa$ has higher variation due to pure round-off noise (i.e., the mean reversion speed is the most sensitive numerically).
We therefore conclude that the recalibration of the fast factorization is statistically indistinguishable from retraining fully, indicating perfect numerical agreement and no systematic bias across the parameters.
While we are aware of the fact that we provide an exemplary comparison, we remark that in several hundreds of trials based on different sources of randomness; we failed to see different behavior than that of in Figure~\ref{plot:precisions1}\footnote{We present both estimates with the same color dots given that the dispersion is almost non-existent under the equivalence of machine precision.}.
Unlike the results presented in Figure~\ref{plot:precisions1}--\ref{plot:precisions2} describes an exemplary comparison of all parameters and $\kappa$ for a larger sample.
Especially, in Subfigure~\ref{fig:precb1}, we see more regular behavior due to the possible reasons we discussed in an earlier section,~\ref{sec:fast-refactor}.
Even though we, for the time being, exclude the sharded recomputation for brevity in Figures~\ref{plot:precisions1}--\ref{plot:precisions2}; similar visualizations could be constructed on the exactness of the sharded recomputation.
We remark, however, that there is no difference between all three approaches up to machine precision, i.e., $\|\theta_{\text{fast}}-\theta_{\text{retrain}}\|_2 < 10^{-8}  $ and $\|\theta_{\text{recomp}}-\theta_{\text{retrain}}\|_2 < 10^{-8} $  in all instances we observed in large sample experiments along with.
Another reason for the exclusion of the sharded computation is timely, as it is highly sensitive to the forgetting set which we now discuss in Figure~\ref{plot:deg1}.

\begin{figure}[H]
	\begin{subfigure}{0.5\textwidth}
		\centering
		\includegraphics[ width=1\linewidth]{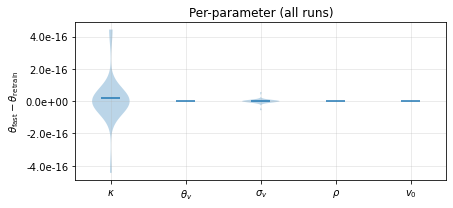}
		\caption{ Per parameter (all runs) }
		\label{fig:precb1}
	\end{subfigure}
	\begin{subfigure}{0.5\textwidth}
		\centering
		\includegraphics[width=1\linewidth]{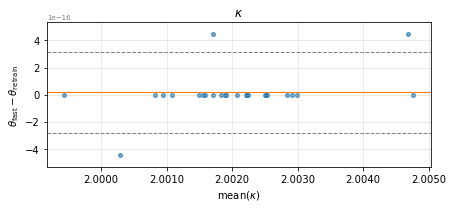}
		\caption{$\kappa$}
		\label{fig:precb2}
	\end{subfigure}
	\caption{An exemplary comparison of equivalence of retraining and fast factorization with a larger sample  }
	\label{plot:precisions2}
\end{figure}

Full recalibration scales  with the data set size since each quote requires multiple finite-difference Heston evaluations.
In contrast, the proposed fast refactor requires no re-pricing and runs in sub-millisecond time, as it merely updates the cached curvature system, unlike the sharded recomputation approach for reason we discuss now.
Recall that $F$ refers to the forget set, i.e. the subset of quotes the user asked to \emph{unlearn}, $\mathcal{K}(F)$ refers to the set of affected shards with $K$ being the total number of shards.
The sharded recomputation should be faster than retraining only if the number of affected shard is smaller than number of shards, $\mid \mathcal{K}(F ) \mid \ll K $. 
Therefore, two cases become interesting that either there exists a small number of shards, or with higher likelihood that the forget set is spread \emph{roughly uniformly} across all shards.
The second case suggests that almost every shard is affected.
In small experiments (few data per shard), shard recomputation skips a few shards and is faster.
In full-scale runs, when almost every shard contains forgotten quotes, recomputation is likely to degenerates to full retraining.
The recomputation cost approaches full retraining when the forget set is evenly distributed across shards. 
Only the fast refactorization operator retains sub-millisecond latency regardless of shard coverage.

\begin{figure}[H]
	\begin{subfigure}{0.5\textwidth}
		\centering
		\includegraphics[ width=1\linewidth]{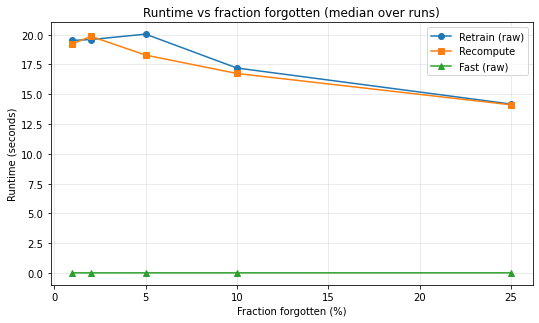}
		\caption{Larger-sample configuration }
		\label{fig:recompequalalmost1seed42}
	\end{subfigure}
	\begin{subfigure}{0.5\textwidth}
		\centering
		\includegraphics[width=1\linewidth]{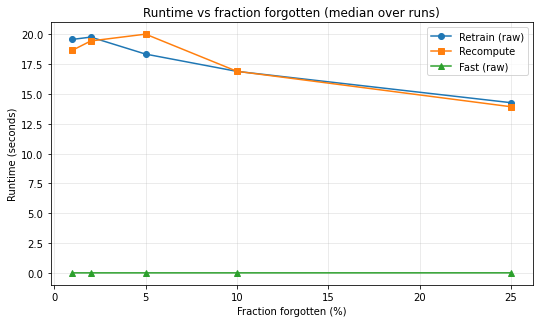}
		\caption{Larger-sample configuration}
		\label{fig:recompequalalmost1seed60}
	\end{subfigure}
	\begin{subfigure}{0.5\textwidth}
		\centering
		\includegraphics[width=1\linewidth]{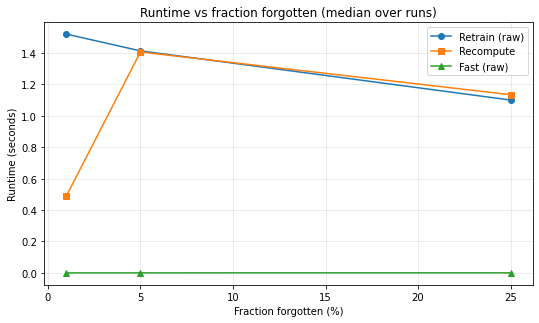}
		\caption{Smaller-sample configuration}
		\label{fig:recompequalalmost1seed42_d}
	\end{subfigure}
	\begin{subfigure}{0.5\textwidth}
		\centering
		\includegraphics[width=1\linewidth]{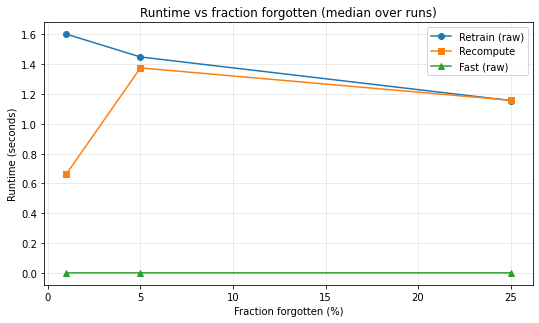}
		\caption{Smaller-sample configuration}
		\label{fig:recompequalalmost1seed60_d}
	\end{subfigure}
	\caption{Example on the importance of the forgetting set for the sharded recomputation}
	\label{plot:deg1}
\end{figure}
We show this aspect in Figure~\ref{plot:deg1}, in which Subfigures~\ref{fig:recompequalalmost1seed42}--\ref{fig:recompequalalmost1seed60} are based on a larger sample with the fixed random forgetting but different sources of randomness in the underlying paths leading to different quotes, the ones on the lower panel are from a smaller sample with the same characteristics; hence much lower computational time given in Subfigures~\ref{fig:recompequalalmost1seed42}--\ref{fig:recompequalalmost1seed60}.
In each example observable in Figure~\ref{plot:deg1}, the sharded recomputation tends to be in a co-movement with the retraining fully, especially after higher percentage of unlearned quotes.  
So that in case of worst-case dispersion of deletions, i.e. too many shards being affected, the sharded recomputation is no longer cheap in computational cost.

\begin{figure}[H]
	\begin{subfigure}{0.5\textwidth}
		\centering
		\includegraphics[ width=1\linewidth]{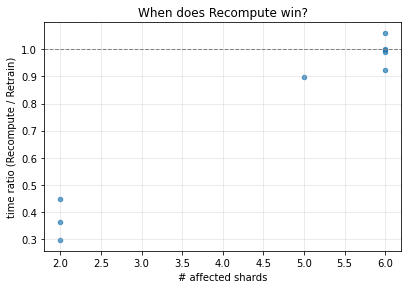}
		\caption{Effect of   earlier shards }
		\label{fig:when1}
	\end{subfigure}
	\begin{subfigure}{0.5\textwidth}
		\centering
		\includegraphics[width=1\linewidth]{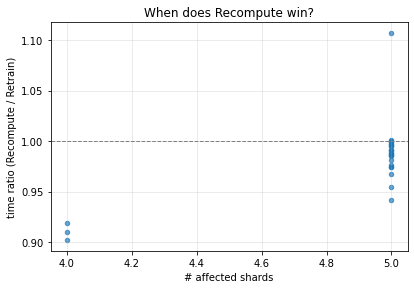}
		\caption{Effect of no earlier shards }
		\label{fig:when2}
	\end{subfigure}
	\caption{Exemplary demonstration of the sharding positions in sharded recomputation   }
	\label{plot:whens}
\end{figure}

In Figure~\ref{plot:whens}, we demonstrate different forgetting parts and its effect on the relative computation time of the sharded recomputation by retraining. 
Time ration on the $y$-axis shows that if the time ratio is lower than one, the sharded recomputation is faster than retraining on the retained data.
By construction, $x$-axis shows how many shards contain at least one forgotten quote.
While the second shard affected in Subfigure~\ref{fig:when1}, Subfigure~\ref{fig:when2} illustrates the first three shard is not affected.
If we were to decide not to unlearn some of the quotes in the experiments, this could overall be efficient in the first example, and with likely none to little effect on the second example in Subfigure~\ref{fig:when2} as the current set-up is saving so little time given the affected shards.
Therefore, we stress strongly that the sharded recomputation method's efficiency is governed by the locality of the forgotten data.

Across multiple independent runs, the relative parameter error consistently remained at machine precision, confirming the generality of the result.
The surface is fixed intentionally to ensure that performance differences stem solely from the unlearning mechanism, not from new random draws. 
Additional seeds produced qualitatively identical behavior
The reported numbers in Tables~\ref{tab:benchmark1} and~\ref{tab:benchmark2} are medians across 10 runs and variability was negligible (IQR\footnote{IQR stands for interquartile range.} below 1$\%$).
While $s$ refers to seconds, $ms$ refers to milliseconds ($10^{-3}s$) and $\mu s $ refers to microseconds ($10^{-6}s$).

\begin{table}[h]
\centering
 \caption{%
	Benchmark results across forgetting fractions.
	Median runtimes   and parameter deviations are reported.
	The fast-refactor approach achieves identical accuracy to full retraining
	with several orders of magnitude speedup.
}
 \begin{tabular}{rccccc}
 	\hline
 	\textbf{F(\%)} &
 	\textbf{Retrain} &
 	\textbf{Recompute} &
 	\textbf{Fast} &
 	\textbf{RMSE kept (fast/retr)} &
 	\textbf{Speedup} \\
 	\hline
 	1$\%$  & 21.06 s & 20.44 s & 174.2  $\mu s$ & 0.00049 / 0.00049 & $\times$120,886.6 \\
 	2$\%$  & 18.36 s & 18.05 s & 242.8 $\mu s$ & 0.00049 / 0.00049 & $\times$75,639.3 \\
 	5$\%$  & 17.71 s & 17.53 s & 436.9  $\mu s$ & 0.00049 / 0.00049 &$\times$40,414.3 \\
 	10$\%$ & 16.79 s & 16.61 s & 882.0 $\mu s$ & 0.00049 / 0.00049 & $\times$19,268.7 \\
 	25$\%$ & 14.31 s & 13.93 s & 2.12  $m s$  & 0.00050 / 0.00050 & $\times$6,565.9 \\
 	\hline
 \end{tabular}

 \label{tab:benchmark1}
\end{table}

Although in Tables~\ref{tab:benchmark1}--\ref{tab:benchmark2}, we employ different underlying paths so that different quotes could be generated, yet the forgetting sets are fixed across runs.
The first columns refer the fraction of data forgotten, and throughout the study we never control the forgetting set; we simply randomize it.
Reported values show that, on general, even in the worst case scenario the fast factorization speeds up the calibration by six thousand times, roughly four to five orders of magnitude speedup remarking strong evidence of the efficiency of the unlearning operator.
We remark that the speedup is measured as the ratio of median computation time taken via retraining by the median computation time taken via the fast factorization operator. 
The fifth columns in Tables~\ref{tab:benchmark1}--\ref{tab:benchmark2}  measure the validation error on retained quotes, 
 

\begin{table}[http]
	\centering
	 \caption{%
		Benchmark results across forgetting fractions.
		Median runtimes and parameter deviations are reported.
		The fast-refactor approach achieves identical accuracy to full retraining
		with several orders of magnitude speedup.
	}
\begin{tabular}{rccccc}
	\hline
	\textbf{F(\%)} &
	\textbf{Retrain} &
	\textbf{Recompute} &
	\textbf{Fast} &
	\textbf{RMSE kept (fast/retr)} &
	\textbf{Speedup} \\
	\hline
	1$\%$  & 18.16 s & 17.63 s & 154.0 $\mu s$ & 0.00049 / 0.00049 & $\times$120,110.6 \\
	2$\%$  & 18.01 s & 17.42 s & 230.0 $\mu s$ & 0.00049 / 0.00049 & $\times$78,171.0 \\
	5$\%$  & 17.13 s & 17.08 s & 444.2 $\mu s$ & 0.00049 / 0.00049 & $\times$38,550.8 \\
	10$\%$ & 16.19 s & 16.17 s & 816.2 $\mu s$ & 0.00049 / 0.00049 & $\times$20,276.0 \\
	25$\%$ & 13.39 s & 13.37 s & 1.90 ms  & 0.00050 / 0.00050 & $\times$7,059.7 \\
	\hline
\end{tabular}

\label{tab:benchmark2}
\end{table}

Tables~\ref{tab:benchmark1}--\ref{tab:benchmark2} illustrate  retraining time does not scale perfectly linearly (since GN step cost flattens with fewer quotes).
Therefore, denominator decreases slightly faster than numerator.
We also observe that the effective speedup decreases monotonically with the forgotten fraction since a larger fraction of the cached structure must be updated or recomputed.
This scaling is consistent with the theoretical expectation that the cost advantage of refactorization diminishes as the retained set shrinks.
All timings were measured as median wall-clock durations over 10 runs using identical random seeds and quote subsets.
We now provide visual presentation of the experiments reported in Tables~\ref{tab:benchmark1}--\ref{tab:benchmark2}.


\begin{figure}[H]
	\begin{subfigure}{0.5\textwidth}
		\centering
		\includegraphics[ width=1\linewidth]{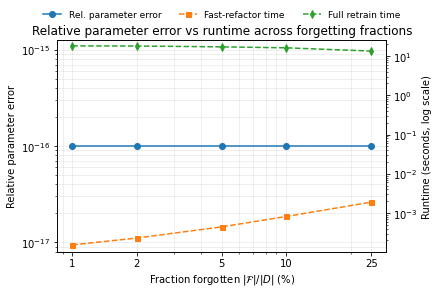}
		\caption{Typical behavior }
		\label{fig:rel1}
	\end{subfigure}
	\begin{subfigure}{0.5\textwidth}
		\centering
		\includegraphics[width=1\linewidth]{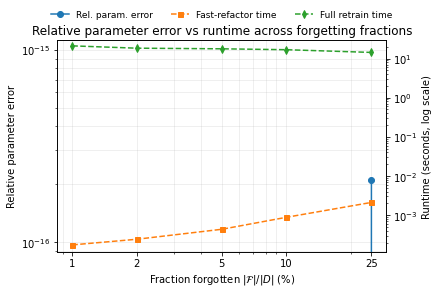}
		\caption{Slight differencing at machine precision }
		\label{fig:rel2}
	\end{subfigure}
	\caption{Exemplary comparison of relative parameter error and runtime across forgetting fractions  }
	\label{plot:rels}
\end{figure}

Figure~\ref{plot:rels} demonstrates relative parameter error against runtime across various forgetting fractions.
In Subfigures~\ref{fig:rel1}--\ref{fig:rel2}, the deviation is calculated by $\|\theta_{fast}-\theta_{retrain}\|_2/\|\theta_{retrain}\|_2$ given $D\setminus F$.
In Subfigure~\ref{fig:rel1}, we observe a perfectly flat blue line, all around $10^{-15}$, whereas Subfigure~\ref{fig:rel2} demonstrates that the blue line jumps upward near $25 \%$ around  $10^{-12}$.
That single-point rise indicates that, for one run at the highest forgetting fraction fast and retrain parameters differed slightly more, probably due to numerical conditioning or cache subtraction noise.
Although this is evident that the unlearning operator behaves stably as relative parameter deviation remains near machine precision; an exemplary case such as this requires several remarks.
Higher percentage removals are actually a loss of too many informative points on the surface, carrying the risk of $H'$ becoming poorly conditioned and amplifying small floating-point noise in $G'$.
Even with that, we remark that  fast refactor method reproduces the retraining solution to floating-point accuracy.

A crucial remark is that we empirically observe that the positive-definiteness condition  $\lambda_{\min}(H')>0$ remains satisfied well beyond typical forgetting levels.
In particular, the theoretical bound, $\|\sum_{i\in F}\psi_i\|_2<\lambda_{\min}(H)$, ensures that $H'$ remains positive definite, is rarely active until more than approximately $70\%$ of quotes are removed.
This indicates a strong numerical robustness of the downdate procedure and supports the stability of the fast-refactor updates under realistic unlearning scenarios.
We deem a further examination of the eigenvalue structure of the curvature matrix $H'$   
after forgetting to be of deeper interest to optimization algorithms rather than to the methodological structure of our unlearning framework. 
Therefore, we keep this discussion brief and do not pursue an extensive numerical analysis  beyond stability verification.

Empirically, calibration runtimes scale linearly in the number of quotes $N$ 
and approximately quadratically in the number of Fourier--Simpson integration 
nodes $N_u$, in line with the overall $\mathcal{O}(N\,N_u)$ cost of evaluating 
the Heston pricing integral.
Since the parameter dimension $p$ is fixed and relatively small, the memory and computational cost of assembling and solving  the Gauss--Newton normal equations is negligible. 
The reported wall--clock times therefore match the analytic complexity of the pricing integral and the minimal  number of nonlinear iterations typically required for Heston calibration.

The goal of our benchmark is not large--scale industrial calibration but a  controlled and reproducible comparison of unlearning operators.
In this setting, second-level runtimes are representative and analytically consistent. The key 
question is whether the proposed operators reproduce the calibration update on 
the reduced dataset. Both unlearning operators achieve numerical agreement up to 
machine precision, and the fast refactorization operator does so at a small 
fraction of the computational cost; see Proposition~\ref{prop:complexity}.

\section{Conclusion}

We have shown that the proposed unlearning operators admit rigorous guarantees; local exactness under fixed linearization and stability under curvature perturbations
Numerical experiments further confirm that the fast refactorization operator matches full retraining to floating-point precision and achieves several orders  of magnitude speedup, even for substantial forgetting fractions. 
Although algebraically simple, these operators rely on the observation that the Gauss--Newton normal equations encode a sufficient-statistics structure for nonlinear calibration.
Identifying the exact additive quantities whose removal preserves the optimality conditions under fixed linearization is, in our view, a nontrivial contribution and appears not to have been articulated previously in either the calibration or unlearning literature.

From a computational perspective, recalibration of the Heston model involves the evaluation of semi-analytic Fourier--Simpson integrals and accumulation of the $p\times p$ curvature matrix $H=\sum_i J_i^\top J_i$. 
Even in the full configuration ($U_{\max}=120$, $N_u=800$), this entails only $N\times N_u \approx 10^6$ function evaluations for datasets of typical size ($N \approx 10^3{-}10^4$).
Such runtimes are modest in isolation, but financial institutions routinely process thousands of option books or parameter updates per day.
Recomputing all normal equations after each deletion therefore becomes costly, whereas the proposed unlearning operators perform mathematically exact deletions using only cached curvature statistics. The contribution is thus not  raw speed alone, but the ability to \emph{delete data deterministically without 	retraining}, enabling reversible and auditable calibration updates at negligible incremental cost.

Beyond computational gains, the framework reframes recalibration as an additive-subtractive operator calculus, enabling principled deletion of corrupted, stale, or restricted data. 
This expands calibration from a purely forward-learning procedure into a bidirectional model-management process, useful for regulatory compliance, data-quality control, and influence diagnostics. 
Although it is not in our interest yet our framework could also be used to quantify the influence of the subsets of data.
Questions such as which period affecting the calibration most, or what happens in case of exclusion of a data source might also be asked.

We focus on one representative semi-analytic model to isolate the operator behavior; extension to other models and real data is left for future work.
Future work may explore whether the operator perspective developed here extends beyond the static Gauss--Newton setting.
One natural question is how  additive-subtractive updates interact with models that contain latent or filtered state variables, such as stochastic-volatility or regime-switching specifications, where forgetting would couple to the underlying filter-smoother structure. 
Another possible direction concerns the use of the resulting sufficient-statistics calculus for influence diagnostics, for example to quantify the leverage of particular maturities, regimes, or data sources on the calibrated parameters. 
The same viewpoint also suggests potential analogues for rolling-window or online calibration procedures, where the removal of stale information must be carried out without repeatedly rebuilding the normal equations from scratch.
Finally, a more ambitious line of inquiry is whether analogous operator rules exist for quasi-Newton or higher-order curvature representations.
Overall, these directions emphasize that machine unlearning should be regarded not only as a computational device, but as part of a broader operator-theoretic framework for interpretable, auditable, and dynamically maintainable calibration pipelines.

\backmatter


\section*{Compliance with Ethical Standards}


\bmhead*{Competing Interests}
 There are no financial or non-financial interests   directly or indirectly related to the work submitted for publication.
\bmhead*{Funding}
 There is no funding received in the making of this manuscript.




\bigskip

\begin{appendices}






\end{appendices}


\bibliography{MyBiblio}

@article{ginart2019making,
  title={Making ai forget you: Data deletion in machine learning},
  author={Ginart, Antonio and Guan, Melody and Valiant, Gregory and Zou, James Y},
  journal={Advances in neural information processing systems},
  volume={32},
  year={2019}
}

@article{zhang2023review,
  title={A review on machine unlearning},
  author={Zhang, Haibo and Nakamura, Toru and Isohara, Takamasa and Sakurai, Kouichi},
  journal={SN Computer Science},
  volume={4},
  number={4},
  pages={337},
  year={2023},
  publisher={Springer}
}

@article{nguyen2025survey,
  title={A survey of machine unlearning},
  author={Nguyen, Thanh Tam and Huynh, Thanh Trung and Ren, Zhao and Nguyen, Phi Le and Liew, Alan Wee-Chung and Yin, Hongzhi and Nguyen, Quoc Viet Hung},
  journal={ACM Transactions on Intelligent Systems and Technology},
  volume={16},
  number={5},
  pages={1--46},
  year={2025},
  publisher={ACM New York, NY}
}

@article{sekhari2021remember,
  title={Remember what you want to forget: Algorithms for machine unlearning},
  author={Sekhari, Ayush and Acharya, Jayadev and Kamath, Gautam and Suresh, Ananda Theertha},
  journal={Advances in Neural Information Processing Systems},
  volume={34},
  pages={18075--18086},
  year={2021}
}

@article{qu2024learn,
  title={Learn to unlearn: Insights into machine unlearning},
  author={Qu, Youyang and Yuan, Xin and Ding, Ming and Ni, Wei and Rakotoarivelo, Thierry and Smith, David},
  journal={Computer},
  volume={57},
  number={3},
  pages={79--90},
  year={2024},
  publisher={IEEE}
}

@inproceedings{bourtoule2021machine,
  title={Machine unlearning},
  author={Bourtoule, Lucas and Chandrasekaran, Varun and Choquette-Choo, Christopher A and Jia, Hengrui and Travers, Adelin and Zhang, Baiwu and Lie, David and Papernot, Nicolas},
  booktitle={2021 IEEE symposium on security and privacy (SP)},
  pages={141--159},
  year={2021},
  organization={IEEE}
}

@article{guo2019certified,
  title={Certified data removal from machine learning models},
  author={Guo, Chuan and Goldstein, Tom and Hannun, Awni and Van Der Maaten, Laurens},
  journal={arXiv preprint arXiv:1911.03030},
  year={2019}
}

@article{heston1993closed,
  title={A closed-form solution for options with stochastic volatility with applications to bond and currency options},
  author={Heston, Steven L},
  journal={The review of financial studies},
  volume={6},
  number={2},
  pages={327--343},
  year={1993},
  publisher={Oxford University Press}
}

@article{dumas1998implied,
  title={Implied volatility functions: Empirical tests},
  author={Dumas, Bernard and Fleming, Jeff and Whaley, Robert E},
  journal={The Journal of Finance},
  volume={53},
  number={6},
  pages={2059--2106},
  year={1998},
  publisher={Wiley Online Library}
}

@article{ulrich2023implied,
  title={Implied volatility surfaces: a comprehensive analysis using half a billion option prices},
  author={Ulrich, Maxim and Zimmer, Lukas and Merbecks, Constantin},
  journal={Review of Derivatives Research},
  volume={26},
  number={2},
  pages={135--169},
  year={2023},
  publisher={Springer}
}

@article{homescu2011implied,
  title={Implied volatility surface: Construction methodologies and characteristics},
  author={Homescu, Cristian},
  journal={arXiv preprint arXiv:1107.1834},
  year={2011}
}

@article{friedman2014some,
  title={Some Economically Meaningful Option Model Calibration Performance Measures},
  author={Friedman, Craig A and Cao, Wenbo and Huang, Yuchang},
  journal={Available at SSRN 2193803},
  year={2014}
}

@article{ait1998nonparametric,
  title={Nonparametric estimation of state-price densities implicit in financial asset prices},
  author={A{\"\i}t-Sahalia, Yacine and Lo, Andrew W},
  journal={The journal of finance},
  volume={53},
  number={2},
  pages={499--547},
  year={1998},
  publisher={Wiley Online Library}
}

@article{date2011linear,
  title={Linear and non-linear filtering in mathematical finance: a review},
  author={Date, Paresh and Ponomareva, Ksenia},
  journal={IMA Journal of Management Mathematics},
  volume={22},
  number={3},
  pages={195--211},
  year={2011},
  publisher={OUP}
}

@article{broto2004estimation,
  title={Estimation methods for stochastic volatility models: a survey},
  author={Broto, Carmen and Ruiz, Esther},
  journal={Journal of Economic surveys},
  volume={18},
  number={5},
  pages={613--649},
  year={2004},
  publisher={Wiley Online Library}
}

@article{bakshi1997empirical,
  title={Empirical performance of alternative option pricing models},
  author={Bakshi, Gurdip and Cao, Charles and Chen, Zhiwu},
  journal={The Journal of finance},
  volume={52},
  number={5},
  pages={2003--2049},
  year={1997},
  publisher={Wiley Online Library}
}

@article{broadie2007model,
  title={Model specification and risk premia: Evidence from futures options},
  author={Broadie, Mark and Chernov, Mikhail and Johannes, Michael},
  journal={The Journal of Finance},
  volume={62},
  number={3},
  pages={1453--1490},
  year={2007},
  publisher={Wiley Online Library}
}
\end{document}